\documentclass[a4paper,11pt]{article}

\usepackage{graphicx}
\usepackage{amsfonts}
\usepackage{amssymb}
\usepackage{slashed}
\usepackage{braket}
\usepackage[margin=1.25in]{geometry}
\usepackage{bm}
\usepackage{siunitx}
\usepackage{wrapfig}
\usepackage{mathtools}
\usepackage[standard]{ntheorem}
\usepackage{lipsum}
\usepackage{csquotes}
\usepackage{xcolor}

\usepackage[clean]{revdiff}

\theoremstyle{theorem}

\theoremstyle{empty}
\newtheorem{duplicate}{NameIgnored}

\usepackage[backend=biber]{biblatex}
\AtBeginDocument{

\newcommand{\cls}[1]{$\mathsf{#1}$}
\DeclareSymbolFont{bbold}{U}{bbold}{m}{n}
\DeclareSymbolFontAlphabet{\mathbbold}{bbold}

\DeclareSymbolFont{extraitalic}      {U}{zavm}{m}{it}
\DeclareMathSymbol{\stigma}{\mathord}{extraitalic}{168}
\DeclareMathSymbol{\Stigma}{\mathord}{extraitalic}{167}
}

\hfuzz=\maxdimen \tolerance=10000 \hbadness=10000

\makeatletter
\newcommand{\proofpart}[2]{%
  \par
  \addvspace{\medskipamount}%
  \noindent\emph{Part #1: #2}\par\nobreak
  \addvspace{\smallskipamount}%
  \@afterheading
}
\makeatother

\newcommand{\QCTS}{\cls{Quantum}-\cls{Clock}-\cls{Ternary}-\cls{SAT} }

\title{Quantum Constraint Problems can be complete for \cls{BQP}, \cls{QCMA}, and more}
\author{Alex Meiburg \\ \footnotesize{ameiburg@ucsb.edu}}
\date{%
    \small{University of California, Santa Barbara}\\%
}

\begin{document}
\maketitle

A quantum constraint problem is a frustration-free Hamiltonian problem: given a collection of local operators, is there a state that is in the ground state of each operator simultaneously? It has previously been shown that these problems can be in \cls{P}, \cls{NP}-complete, \cls{MA}-complete, or \cls{QMA_1}-complete, but this list has not been shown to be exhaustive. We present three quantum constraint problems, that are (1) \cls{BQP_1}-complete (also known as \cls{coRQP}), (2) \cls{QCMA}-complete and (3) \cls{coRP}-complete. This provides the first natural complete problem for \cls{BQP_1}. We also show that all quantum constraint problems can be realized on qubits, a trait not shared with classical constraint problems. These results suggest a significant diversity of complexity classes present in quantum constraint problems.

\section{Background}
\subsection{Classical constraint problems}
The classical notion of {\em constraint problem} or {\em constraint satisfaction problem} (CSP) take the form of a {\em domain} of variable values, and {\em clauses}: relations on a set of variables. Usually these terms are specifically in reference to {\em finite} constraint problems, where the domain $D$ is finite, and the clauses $C$ are of bounded arity $k$. A clause of arity $m \le k$ is a subset of $D^k$. We will focus our attention on finite CSPs, referring to them simply as CSPs.
\par Four representative examples could be \cls{2SAT}, \cls{3SAT}, \cls{3COLOR}, and \cls{Mod3} (the set of linear equations in variables mod 3, where each equation has at most 3 variables). A CSP {\em instance} is a finite number of variables $n$, and a list of clauses applied to certain variables. The instance is satisfiable iff there is an assignment $A : [n] \to D$ such that, for every clause $c \in C$ applied to variables $(v_1, \dots v_m)$, the assigned values $(A(v_1),\dots A(v_m)) \in c$. The problem corresponding to a CSP is determining the satisfiabiliy of its instances. \cls{2SAT} and \cls{Mod3} are in \cls{P}, while \cls{3SAT} and \cls{3COLOR} are NP-Complete, that is, in \cls{NPC}.
\par A landmark theorem by Dmitry Zhuk \cite{zhuk}, the so-called CSP Dichotomy Theorem, showed that {\em every} constraint problem is either in \cls{P} or \cls{NPC}. This does not rule out the possibility of other NP-Intermediate languages (the set \cls{NP}$\setminus($\cls{P}$\cup$\cls{NPC}$)$), problems that cannot be expressed as {\bf constraint} problems in particular. Problems such as graph isomorphism or integer factorization are believed to be in this class. A single instance of graph isomorphism (\cls{GI}) can be easily encoded as a single instance of a constraint problem (such as \cls{3SAT}), but then there are other instances of \cls{3SAT} that are much harder than mere graph isomorphism. Finding a constraint problem that {\em only} permitted the expression of \cls{GI} problems would immediately yield either a polynomial time algorithm for \cls{GI}, or a subexponential algorithm for \cls{NPC}, by Zhuk's result. The CSP Dichotomy Theorem also provides a systematically checkable condition for whether a problem is in \cls{P} or \cls{NPC}, the existence of a {\em polymorphism}.
\par A natural question is how this result might translate to the world of quantum problems. We define the quantum version of constraint problems, and emphasize the distinction from quantum {\em optimization} problems. We review known results about quantum constraint problems. The main contribution is providing three new quantum constraint problems, that are complete for the classes \cls{BQP_{1}=coRQP}, \cls{QCMA}, and \cls{coRP}. These imply that any putative quantum dichotomy theorem would need at least 7 distinct cases -- or a proof that some of these 7 complexity classes are actually equal to one another -- in stark contrast to the 2 cases in the classical case.

\subsection{Quantum Constraint Problems}
Quantum Constraint Satisfaction Problems, or QCSPs, can be viewed as a quantum version of a CSP or as a question about frustration free Hamiltonians. A QCSP has a domain size $d$, and a set of clauses or {\em interactions} $C = \{\mathcal{H}_i\}$. A clause $\mathcal{H}_i$ of arity $m$ is a Hermitian operator on the space of $m$ many $d$-qudits, $(\mathbb{C}^d)^{\otimes m}$. We require that each clause is a projector, i.e. $\mathcal{H}_i^2 = \mathcal{H}_i$. For a problem instance on $n$ variables, the interactions $H_i$ extend naturally to operators on the whole Hilbert space. If the arity of all interactions is at most $k$, then the QCSP is {\em $k$-local}. An instance of a QCSP is then a collection of the interactions applied at different qudits. We use fonts to distinguish the clause types of the QCSP $\mathcal{H}_i$, from the particular clauses of the instance $H_i$. An instance is satisfiable if there is a state $\ket{\psi} \neq 0$ such that $\forall_i H_i \ket{\psi} = 0$, equivalently if the total Hamiltonain $H = \sum_i H_i$ satisfies $H \ket{\psi} = 0$.
\par This can be viewed physically as the question, is $H$ frustration-free? Frustration-free Hamiltonians have applications in one-way computation\cite{Chen09}, and are often easier to study in terms of entanglement structure.
\par If there is a satisfying assignment (a YES instance), we expect we should be able to verify this state exactly, and accept with probability one. If the ground state has positive energy bounded from below by an inverse polynomial $1/p(n)$, then running $p(n)$ rounds should suffice to detect the failure with high probability. But if the true ground state has positive but exponentially small energy, we may be unable to observe this small energy, and erroneously accept. Thus this paper will discuss a promise problem variants of QCSPs, which excludes superpolynomially small gaps. We are given a constant $b > 1/poly(n)$, and promised that the either the total Hamiltonian has a frustration-free ground state, or that the ground state has energy at least $b$. All complexity classes mentioned herein are formulated as promise problem classes, with perhaps with trivial promises (\cls{P} and \cls{NP}, which requires no promise), and all completeness theorems refer to promise problem completeness. Probabilistic classes such as \cls{QMA} are semantic classes, not syntactic, which largely precludes the possibility of a non-promise problem being complete; and so virtually all discussion of the classes concerns their promise variants. \cite{Brav09} discusses this distinction in more detail.

\subsection{Previous Results}
\par We review some known statements about quantum constraint problems.
\begin{proposition}[Folklore.]
Every classical CSP can be efficiently mapped to a corresponding QCSP, preserving satisfiability.
\end{proposition}
This occurs by simply writing each classical clause in the classical basis, where they are diagonal.
\begin{definition}[$k$-QSAT \cite{Brav06}]
The QCSP $k$-\cls{QSAT} is the QCSP with $d=2$ qubits, and with clauses $C$ as the set of all $k$-local interactions.
\end{definition}
\begin{theorem}[\cite{Brav06}]
2-\cls{QSAT} is in \cls{P}. 
\end{theorem}

\begin{definition}[\cls{QMA_1} \cite{Goss13}]
A language $L$ belongs to the class \cls{QMA_1} iff there is a uniform family of quantum circuits $U$ of polynomial size, such that for an input $x$:

\noindent \textbf{Perfect Completeness}: If $x\in L$, then there exists a state $\ket{y}$, such that measuring the first qubit of $U\ket{x}\ket{y}$ is 1 with probability 1.

\noindent \textbf{Soundness}: If $x\not\in L$, then for any state $\ket{y}$, such that measuring the first qubit of $U\ket{x}\ket{y}$ is 0 with probability at least $2/3$.
\end{definition}

\begin{theorem}[\cite{Goss13}]
3-\cls{QSAT} is \cls{QMA_1} complete. 
\end{theorem}
\begin{definition}[$(r,s)$-\cls{QSAT}]
The QCSP $(r,s)$-\cls{QSAT} is most naturally described as having a mixture of $r$-qudits and $s$-qudits, with 2-local clauses only between $r$- and $s$-qudits. This can be defined in our definition of QCSP, using $d = r+s$, where allowed clauses are 2-local projectors that project onto the first $r$ states in the first qudit, and the last $s$ states in the second qudit. 
\end{definition}
Note that $(2,2)$-\cls{QSAT} is the same as 2-\cls{QSAT}, and that $(r',s')$-\cls{QSAT} is at least as hard $(r,s)$-\cls{QSAT} if $r' \ge r,\, s'\ge s$.
\begin{theorem}[\cite{Elda08}]
(3,5)-\cls{QSAT} is \cls{QMA_1} complete.
\end{theorem}
It is worth also noting that the analogous (2,3)-\cls{SAT} is \cls{NP}-complete, but it is unknown if (2,3)-\cls{QSAT} is \cls{QMA_1} complete or not. Some 2-local qutrit Hamiltonians such as the AKLT Hamiltonian\cite{Affl87} exhibit interesting entanglement structure in a frustration free system.
\begin{definition}[$k$-\cls{STOQ}-QSAT \cite{Brav09}]
The QCSP $k$-\cls{STOQ}-\cls{QSAT} is the version of $k$-\cls{QSAT} restricted so that all clauses have nonpositive off-diagonal elements (they are ``stoquastic'').
\end{definition}
\begin{theorem}[\cite{Brav09}]
For any $k \ge 6$, the problem $k$-\cls{STOQ}-\cls{QSAT} is \cls{MA}-complete. For $k < 6$, it is contained in \cls{MA}.
\end{theorem}

\subsection{Quantum Optimization Problems}
As an aside, it is worth pointing out the difference between quantum constraint problems and quantum {\em optimization} problems. A quantum optimization problem has a similar form: an allowed set of local operators, and we ask if there is a a quantum state with sufficiently low energy. The two key differences are that (1) the local operators are not necessarily projectors, and (2) we ask if the ground state has energy below some $a$, as opposed to being in the ground state of all clauses simultaneously. The Hamiltonians constructed in satisfiable instances are not necessarily frustration free, then.
\par The complexity classes of these problems were classified\footnote{Technically, this result was shown only for problems on qubits; and then, only for when the clauses are all 2-local, or all Pauli matrices are allowed clauses. It seems reasonable to expect they generalize.} by Cubitt and Montanaro in \cite{Cubi16}, where it was shown they fall into \cls{P}, \cls{NP}-complete, \cls{StoqMA}-complete, and \cls{QMA}-complete. Indeed, it is known the optimization problem 2-\cls{Local}-\cls{Hamiltonian} is already \cls{QMA}-complete\cite{Kemp06}, analogous to \cls{MAX}-2-\cls{SAT} already being \cls{NP}-complete. Cubitt and Montanaro's classification is analogous to the ``Min CSP classification theorem'' on classical optimization problems, an optimization-oriented analog of the dichotomy theorem.  We will not further discuss optimization problems here.

\section{Statement of results}
It is known that some CSPs are not simply in \cls{P}, but in fact complete for \cls{P}: they model all efficient classical computation. It seems natural to ask whether there are QCSPs that capture efficient quantum computation, viewed as \cls{BQP}. The class \cls{BQP} consists of the problems for which there exists uniform quantum circuits of polynomial size that return the correct answer with probability at least $2/3$. However, we are discussing exact satisfying assignments of a Hamiltonian, so allowing a solution that may be wrong $1/3$ of the time would prove very difficult, and we must restrict ourselves to one-sided error. This is the same reason that \cite{Goss13} uses \cls{QMA_1} over \cls{QMA}. We define the class as following:
\begin{definition}[\cls{BQP_1}]
A language $L$ belongs to the class \cls{BQP_1} iff there is a uniform family of quantum circuits $U$ of polynomial size, such that for an input $x$:

\noindent \textbf{Perfect Completeness}: If $x\in L$, then measuring the first qubit of $U\ket{x}$ gives 1 with probability 1.

\noindent \textbf{Soundness}: If $x\not\in L$, then measuring the first qubit of $U\ket{x}$ gives 0 with probability at least $2/3$.
\end{definition}

It is worth noting that \cls{BQP_1} could also be called \cls{coRQP}, the set of complements to \cls{RQP}: quantumly solvable problems with perfect soundness and bounded-error completeness\cite{Bert94}\cite{Bera18}. Our first main result is

\begin{theorem}
\label{thm:BQP1}
There is a fixed set $\mathcal{C}$ of 5-local projectors on 13-dimensional qudits, such that the QCSP for $\mathcal{C}$ is complete for \cls{BQP_1}.
\end{theorem}

The majority of the work shall be in proving this theorem. In the process of constructing these projectors, it will become apparent that two more complexity classes could be handled as well, \cls{QCMA_1} and \cls{coRP}.

\begin{definition}[\cls{QCMA_1}]
A language $L$ belongs to the class \cls{QCMA_1} iff there is a uniform family of quantum circuits $U$ of polynomial size, and a polynomial $p(n)$, such that for an input $x$:

\noindent \textbf{Perfect Completeness}: If $x\in L$, there is a classical bitstring $y \in \{0,1\}^{p(|x|)}$, such that measuring the first qubit of $U\ket{x}\ket{y}$ gives 1 with probability 1.

\noindent \textbf{Soundness}: If $x\not\in L$, then for all $y \in \{0,1\}^{p(|x|)}$, measuring the first qubit of $U\ket{x}\ket{y}$ gives 0 with probability at least $2/3$.

\end{definition}

\begin{definition}[\cls{coRP}]
A language $L$ belongs to the class \cls{coRP} iff there is a random Turing machine $T$, such that $T$ always runs in polynomial time, and for any input $x$:

\noindent \textbf{Perfect Completeness}: If $x\in L$, $T$ accepts $x$ with probability 1.

\noindent \textbf{Soundness}: If $x\not\in L$, $T$ rejects $x$ with probability $\ge 2/3$.

\end{definition}

The class \cls{QCMA_1} is a one-sided error variant of \cls{QCMA}\cite{AH02}. In\cite{qcma1qcma} it was shown that \cls{QCMA_1=QCMA}, so our \cls{QCMA_1}-complete problem is equivalently \cls{QCMA}-complete. \cls{coRP}, and its more common complement class \cls{RP}, are standard\cite{randomized}. They are the one-sided error versions of \cls{BPP}. Our other two main results are,

\begin{theorem}
\label{thm:QCMA1}
There is a fixed set $\mathcal{C}$ of 5-local projectors on 15-dimensional qudits, such that the QCSP for $\mathcal{C}$ is complete for \cls{QCMA}.
\end{theorem}
\begin{theorem}
\label{thm:coRP}
There is a fixed set $\mathcal{C}$ of 5-local projectors on 15-dimensional qudits, such that the QCSP for $\mathcal{C}$ is complete for \cls{coRP}.
\end{theorem}

This seems to be the first discussion of \cls{BQP_1}, and its difficulty is likely similar to that of \cls{BQP}. In the context of \cls{coRP}, the distinction between promise problems and decision problems is more important, and so it is more correct to say we have a \cls{Promise\,coRP}-complete problem.

\par Our last results establish that all QCSPs can be described entirely on qubits, something which is not expected to hold in the classical setting. This implies that there are also QCSPs for \cls{BQP_1}, \cls{QCMA}, and \cls{coRP} on qubits alone.

\begin{theorem} (Informal)
\label{thm:q2Informal}
Every QCSP $\mathcal{C}$ on $d$-qudits is equivalent in difficulty to some QCSP $\mathcal{C}'$ on qubits.
\end{theorem}

\section{Construction Techniques}
Before diving into the construction of the \cls{BQP_1}-complete Hamiltonian, we go over a few of the tools.

\subsection{Universal Gate Set}
In their proof that Quantum 3-SAT is \cls{QMA_1}-Complete, Gosset and Nagaj\cite{Goss13} use a gate set $\mathcal{G_0} = \{\hat H,T,CNOT\}$. This has the property that every matrix element is of the form $(a+ib+\sqrt{2}c+i\sqrt{2}d)/4$, see their Definition 3. For reasons that will become clear later, we would like a gate set that leaves no classical basis state unchanged, and so we use the modified gate set
$$\mathcal{G} = \{\hat H,\hat HT,(\hat H\otimes \hat H)CNOT\}$$
Through appropriate multiplications by $\hat H$, this allows the construction of all of $\mathcal{G_0}$, so this is still universal. It is straightforward to compute all matrix elements and check they have the same form.
A language $L$ is \cls{BQP}$_1$-hard if there is a (classical) polynomial time reduction from any $L' \in$ \cls{BQP}$_1$ to $L$. A language is \cls{BQP_1}-complete if it is both \cls{BQP_1}-hard and in \cls{BQP_1} itself. It is immediate that \cls{EQP} $\subseteq$ \cls{BQP_1} $\subseteq$ \cls{BQP}. Given that \cls{EQP} and \cls{BQP} are commonly understood to capture at least some of the power of quantum computing (e.g. strong oracle separations from \cls{P}), \cls{BQP_1} can be understood to be lower- and upper-bounded between these, and so has some essential quantum nature itself.
\par Unfortunately, \cls{BQP_1} is a class with a vanishing error probability (on one-side), the notion of ``universal gate set'' is delicate. The gate that conditionally rotates phase by $\pi/3$, for instance, cannot be built exactly from $\mathcal{G_0}$, although it can be approximated exponentially well. This is a problem faced in Gosset and Nagaj\cite{Goss13} as well.
\par For this reason, \cls{BQP_1} is not a well-defined complexity class on its own; it requires an assumption on the allowed gate set that the circuit $U$ should be built from. If we fix the gate set $\mathcal{G}$ above, we get a class \cls{BQP_{1,\mathcal{G}}}, of problems with one-sided error solvable using $\mathcal{G}$. The QCSP we will construct below is complete for the class \cls{BQP_{1,\mathcal{G}}}. If we preferred a different gate set $\mathcal{G'}$, say be adding $R(\pi/3)$ gate to $\mathcal{G}$, we would have a new class \cls{BQP_{1,\mathcal{G'}}}, and the problem we construct is complete for that class. In this sense, the construction is generic, and we simply write that the QCSP is \cls{BQP_1}-complete. In section 7 we give a definition of {\em weak} QCSPs that attempts to fix these irritating details.

\subsection{Ternary Logic}
We will use ternary logic (also known as ``dual rail logic''), a standard tool in proving \cls{P}-completeness. While 2-\cls{SAT} is a problem that can be solved efficiently on a classical computer, it is not believed to capture the full power of classical computing. But another problem, \cls{Horn}-\cls{SAT}, is \cls{P}-complete: it captures the full power of classical computation. The ideas in that proof will be important in our construction, lifted to a quantum setting, so we briefly outline it here.

\par \cls{Horn}-\cls{SAT} is the boolean CSP which allows any OR clauses in at most 3 variables with at most one negative variable: clauses like $(v_i \vee v_j),\,(\neg v_i \vee v_j \vee v_k),\,(v_i),\,(\neg v_i)$, but not $(\neg v_i \vee \neg v_j)$. By setting all variables to true, all clauses are satisifed except those of form $(\neg v_i)$, which imply $v_i$ must be false. This propagates, possibly reducing some $(v_i \vee \neg v_j) \to (\neg v_j)$, flipping more variables to false until a satisfying instance or conflict is found. Thus, it is in \cls{P}.
 
\par To be \cls{P}-complete, a CSP must be flexible enough to establish an arbitrary computational graph, usually as a circuit. And yet, it must not permit the construction of problems that require guessing the input, such as a circuit with the input left blank and the output forced ``True'', because this would lead to \cls{NP}-hardness. Informally, \cls{Horn}-\cls{SAT} accomplishes P-completness through {\em dual-rail logic}: for each Boolean variable $v_i$ in the original circuit, we make {\em two} variables in the Horn-SAT instance, $v_{i,T}$ and $v_{i,F}$, representing the assertion that $v_i$ is true or false, respectively. We add the constraints that $v_{i,T} \implies \neg v_{i,F}$. This implies that at most one of them can be true -- but does not rule out the case that both are false. This is the third state, ``Undefined'', of the ternary logic. All logic gates of the circuit are then constructed so that they are trivially satisfied, if both $v_{i,T}$ and $v_{i,F}$ are false. A Horn-SAT problem can then constrain the inputs to a circuit by requiring one or the other variable to be true, but if the input is left unconstrained (as in the NP proof-checking setup), then the circuit can be satisfied by just leaving both $v_{i,T}$ and $v_{i,F}$ false.
\par If our variables are 3-state (instead of Boolean), we can construct P-Complete problems more easily. The dual-rail variable is replaced with a single variable, whose states are labelled ``true'', ``false'', and ``undefined''. The output of a gate with ``undefined'' at any input can be anything, but if a gate has two defined inputs, it must compute its output appropriately. This allows the implication of variable states to only travel forward in the computation graph (conducting computation along the way) and not backwards (which would allow the construction of \cls{NP}-hard input-guessing problems).

\subsection{Monogamy}
The monogamy of entanglement\cite{Coff00} states that a subsystem $A$ cannot be fully entangled with subsystem $B$ and with subsystem $C$ at the same time. A 2-qubit clause such as 
$H_{BellPair} = I-(\ket{00}+\ket{11})(\bra{00}+\bra{11})$
has a unique ground state, a Bell pair, which is fully entangled between its two qubits. In a QCSP with $H_{BellPair}$ applied once to qubits 1 and 2, and applied a second time to qubits 2 and 3, we could immediately reject: any satisfying assignment would require qubit 1 to be fully entangled with qubit 2 and qubit 3, violating monogamy. This is a trick we can use to force certain clauses to pair up certain variables, without allowing the creation of any more complicated constraint graphs.

\section{Construction of \cls{BQP_1}-complete Problem}
The problem will be built with 13-dimensional qudits. We will soon define the problem in terms of the allowed clauses, but first give a more convenient notation for the 13-dimensional space. There is a subspace spanned by three states, labelled as
$$\ket{0_L}, \ket{1_L}, \ket{U_L}$$
another two states labelled as
$$\ket{0_{EC}}, \ket{1_{EC}}$$
and an eight dimensional subspace which is the tensor product of three 2-dimensional spaces:
$$\mathbb{C}(\ket{0_{CL}}, \ket{1_{CL}}) \otimes \mathbb{C}(\ket{0_{CA}}, \ket{1_{CA}}) \otimes \mathbb{C}(\ket{0_{CB}}, \ket{1_{CB}})$$
We will abuse notation somewhat and write operators such as $\ket{1_{CL}}\bra{0_{CL}}$, which should be understood as shorthand for $\textbf{0}^5 \oplus (\ket{1_{CL}}\bra{0_{CL}} \otimes I^4)$. Here $\textbf{0}^5$ is a 5-dimensional zero operator, on the the first 5 of the 13 states, and $I^4$ is the 4-dimensional identity operator on the $CA$ and $CB$ subspace. This is the most natural extension of $\ket{1_{CL}}\bra{0_{CL}}$, which is an operator on a two-dimensional Hilbert space, to an operator on the 13-dimensional Hilbert space.
\par We will also write equations such as $X = Y_{12} + Z_{23}$, to mean that $X$ is a 3-local operator, built from the 2-local operators $Y$ and $Z$, which act on the first two and last two qudits respectively. That is, $X = Y\otimes I + I\otimes Z$.
\par With this notation, we can define the problem now, after which we'll elaborate on how it was constructed, which will motivate the definition.

\begin{definition}[\QCTS]
\par The problem \QCTS is the quantum constraint problem with 5 clauses. The allowed clauses are $H_{Start}$, $H_{End}$, and for each unitary $U$ in the gate set $\mathcal{G}$, a clause $H_{prop,U}$. Since $|\mathcal{G}| = 3$ above, this is five clauses total.
\par To give expressions for the clauses, first define the 1-local operators,
$$H_{L} = I - \ket{0_L}\bra{0_L} - \ket{1_L}\bra{1_L} - \ket{U_L}\bra{U_L}$$
$$H_{E} = I - \ket{0_{EC}}\bra{0_{EC}} - \ket{1_{EC}}\bra{1_{EC}}$$
$$H_{C} = \ket{0_L}\bra{0_L} + \ket{1_L}\bra{1_L} + \ket{U_L}\bra{U_L} + \ket{0_{EC}}\bra{0_{EC}} + \ket{1_{EC}}\bra{1_{EC}}$$
$$P_D = \ket{0_L}\bra{0_L}+\ket{1_L}\bra{1_L}$$
and the 2-local operator,
$$H_{BP} = I - (\ket{0_{CB}0_{CA}}+\ket{1_{CB}1_{CA}})(\bra{0_{CB}0_{CA}}+\bra{1_{CB}1_{CA}})$$
and for any 2-qubit unitary $U$, define the 2-local $T(U)$ by the product
$$T(U) = U_B(P_D\otimes P_D)$$
where $U_B$ acts on $\ket{0_L}$ and $\ket{1_L}$ the way that $U$ would act on $\ket{0}$ and $\ket{1}$, and the zero operator if either input is something else. That is, mapping $\ket{0}$ and $\ket{1}$ to $\ket{0_L}$ and $\ket{1_L}$ induces an isometry $\mathbb{C}^2 \to \mathbb{C}^{13}$; this extends naturally to operators $U : (\mathbb{C}^2)^2 \to (\mathbb{C}^2)^2$.
\par \rnew{There are three more clauses to define, but they are difficult to write down simply as projectors. In terms of determining frustration-free ground states, we really only care about the fact that our operators are positive semidefinie, and geometry of their kernels, not the energies of any excited states. Any non-projector operator could be "normalized" by adjusting the energies of all excited states to be 1. For example, a 2-qubit operator that requires $H_C$ on the first qubit and $H_E$ on the second qubit could be written $H_{C,1} + H_{E,2}$, and this is shorthand for the normalized projector $I - (I - H_C)\otimes(I - H_E) = H_{C,1} + H_{E,2} - H_{C,1}H_{E,2}$. The following three definitions use this shorthand.}
\par The clauses are $H_{Start}$ and $H_{End}$ are 3-local and defined by
\begin{align}
\label{hstart}
H_{Start} =\,& (I - (\ket{0_{EC}0_{CA}}+\ket{1_{EC}1_{CA}})(\bra{0_{EC}0_{CA}}+\bra{1_{EC}1_{CA}}))_{12}\\
& + ((I - \ket{1_{CL}}\bra{1_{CL}})\otimes (I - \ket{0_{L}}\bra{0_{L}}))_{23}\nonumber\\
& + H_{E,1} + H_{C,2}+ H_{L,3}\nonumber
\end{align}
\begin{align}
\label{hend}
H_{End} =\,& (I - (\ket{0_{EC}0_{CB}}+\ket{1_{EC}1_{CB}})(\bra{0_{EC}0_{CB}}+\bra{1_{EC}1_{CB}}))_{12}\\
& + ((I - \ket{0_{CL}}\bra{0_{CL}})\otimes (I - \ket{0_{L}}\bra{0_{L}} - \ket{U_L}\bra{U_L}))_{23}\nonumber\\
& + H_{E,1} + H_{C,2} + H_{L,3}\nonumber
\end{align}
\vspace{0.3mm}

Each $H_{prop,U}$ is 5-local, defined by
\begin{align}
\label{hprop}
H_{prop,U} =\,&  P_D\otimes P_D\otimes \ket{1_C0_C0_C}\bra{1_C0_C0_C}+I^{\otimes 2}\otimes\ket{1_C1_C0_C}\bra{1_C1_C0_C} \\
& - T(U)\otimes \ket{1_C1_C0_C}\bra{1_C0_C0_C} - T(U)^\dagger\otimes \ket{1_C0_C0_C}\bra{1_C1_C0_C} \nonumber\\
& + H_{BP,34} + H_{BP,45}\nonumber\\
& + H_{L,1} + H_{L,2} + H_{C,3} + H_{C,4} + H_{C,5}\nonumber
\end{align}

\rnew{(These are not projectors as written, and they should be interpreted as the normalized versions; really, that $H_{prop,U}$ is the unique projector with the same kernel as the operator on the right-hand side.)}

\textbf{End definition.}
\end{definition}

Now it is time to give some meaning to the parts of the problem. We can give better names to the 13 basis states. The first three basis vector span the ``logical subspace'', on which we will do ternary logic: $0_L$ and $1_L$ represent logical 0 and 1 qubit states, and $U_L$ represents an ``undefined'' qubit. The operator $H_L$ just requires that a particular 13-qudit is, in fact, a logical qubit.
\par Since our construction is based on Kitaev's circuit-to-Hamiltonian mapping, we need clock qudits, but we will assign them separate states than the logical qudits. The latter 10 vectors of our 13-dimensional space form the ``clock subspace'', which will perform this role. This allows to avoid worrying whether a $\ket{0}$ is a ``clock zero'' or a ``logic zero''. Compare this with Kitaev's Hamiltonian, where the same $\ket{0}$ state is used for both, but many constraint problems can be built that look nothing like a circuit.
\par Eight of the clock states are a tensor product of three 2-dimensional subspaces: $\{\ket{0_{CL}},\ket{1_{CL}}\}$, $\{\ket{0_{CA}},\ket{1_{CA}}\}$, and $\{\ket{0_{CB}},\ket{1_{CB}}\}$. The $_{CL}$ component is ``logical clock'' states, corresponding to the 0 and 1 clock states in Kitaev's clock construction, and these actually carry the information of the current time is in the circuit's evaluation. The $_{CA}$ and $_{CB}$ components do not carry information on timing, and will be used for something else. The operator $H_C$ requires that a 13-qudit is a clock qudit.
\par Since we will be able to do the same clock-to-Hamiltonian mapping, we will have \cls{BQP_1}-hardness. But in order to keep the difficulty within \cls{BQP_1}, we want to avoid building any problem that looks like something {\em other} than a clock-to-Hamiltonian mapping. In particular, things would be very complicated if the chain of clock states branched, instead of forming a linear path of time. We will use monogamy to uniquely pair each clock qudit -- and therefore, each moment in time -- with a {\em unique} predecessor and successor.
\par The $CA$ and $CB$ components are the auxiliary parts of the clock state, used in the monogamy construction. By establishing a Bell pair between the $CA$ component of a 13-qudit $x$ the $CB$ component of another 13-qudit $y$, we are stating that $y$ is the clock qudit immediately following $x$ in time. Since the $CA$ component of $x$ cannot form another Bell pair by monogamy, $x$ now has a uniquely following moment in time. $H_{BP}$ expresses this constraint between two clock qudits.
\subsection{Initializing and terminating}
\par Unlike $QMA$ problems, we do not want to leave the input up to guessing: the circuit should start in the $\ket{0^n}$ state. For that purpose $H_{Start}$ is designed to force a particular qubit to start in the $\ket{0}$ state. Since the start of time is indicated by the first clock qudit being zero, $H_{Start}$ really only needs to say that: either the clock qudit is one, or the logical qudit is zero.
\par But it would be a headache if it $H_{Start}$ was applied somewhere other than the start of time, effectively forcing logical qubits to be zero in the middle of execution. To avoid this mess, we've added to more ``clock endpoint'' states, $\ket{0_{EC}}$ and $\ket{1_{EC}}$. These can entangled with a $_{CA}$ or $_{CB}$ to form a Bell pair, terminating a chain of clock qudits. For the same reasons of monogamy, a single $_{CA}$ can't be entangled with anything else if it's entangled with a $_{EC}$ qubit.
\par So, $H_{Start}$ can be attached a clock qudit $x$, and require that $x$'s $_{CA}$ subspace is entangled with an $_{EC}$ subspace; this means that $x$ must be the first clock qudit in the chain. Then $H_{Start}$ says that, if $x$ is zero, the logical qubit $y$ must be zero as well. This accomplishes initialization of the input.
\par Breaking down \eqref{hstart}, it is an operator on 3 qudits: an endpoint qudit, a clock qudit, and a logical qudit. The third line forces them to be of these these types. The first line requires that the first two qudits form a maximal Bell pair with their $_{EC}$ and $_{CA}$ states. The second line requires that, when the second (clock) qudit is zero, the third (logical) qudit is also zero.
\par $H_{End}$ does in \eqref{hend} is doing almost the same thing. Instead of pair with $_{CA}$, the endpoint qudit forms a Bell pair with the $_{CB}$ subspace, because we want $H_{End}$ to come at the end of the chain. The second line of \eqref{hend} requires that, when the second (clock) qudit is one, the third (logical) qudit is either $0_L$ or $U_L$ -- zero or undefined. This way, $H_{End}$ states that, at the end of execution, we should not get a ``1'' as a result. If all the input bits are defined, we will get a well-defined ouptut, and $H_{End}$ performs the same role as $H_{out}$ in \cite{Kitaev02} and \cite{Kemp03}.
\subsection{Propagating ternary logic}
\par We are not guaranteed that all of our qudits have an $H_{Start}$ clause on them, and we want to avoid having to guess the input as in QMA problems. In a classical circuit, ternary logic would address this problem as follows: should any of the input bits be uninitialized, all downstream bits can be undefined, leaving the output undefined as well, which is always a satisfying assignment.
\par In our quantum world, where the satisfying state is a superposition of the computing history across all time, there's a simpler solution: just end time itself, and ``destroy the universe'' if we try to compute on an undefined input.
\par In \cite{Kitaev02}, the propagating clause for a two-qubit unitary $U$ was:
$$H_{Kitaev} = I\otimes (\ket{t}\bra{t} + \ket{t-1}\bra{t-1}) - U \otimes \ket{t}\bra{t-1} - U^\dagger \otimes \ket{t-1}\bra{t}$$
$$\textrm{where}\quad \ket{t-1} = \ket{100},\,\, \ket{t} = \ket{110}.$$
For a candidate solution $\ket{\psi}$, it straightforward to check that $\ket{\psi}$ can only be in the nullspace of $H$ if:
$$U\,(I \otimes \bra{t-1})\ket{\psi} = (I \otimes \bra{t})\ket{\psi}$$
that is, the logical state encoded in $\ket{\psi}$ for time $t$ must be equal to $U$ applied to the state at time $t-1$. They must also have the same amplitude. The fact that the amplitudes are equal at each step implies $\ket{\psi}$ must be a uniform superposition across all times.
\par For our $H_{prop,U}$ defined in \eqref{hprop}, we need these ingredients, that for {\em defined} inputs we propagate the state with the same amplitude. If the input is in the undefined state $\ket{U_L}$, we will drop the requirement that the amplitude remains the same. The computation can terminate early, because our solution $\ket{\psi}$ is no longer required to have any component on times after $t-1$.
\par The first line of \eqref{hprop} replaces an $I$ in $H_{Kitaev}$ with a $P_D \otimes P_D$, which projects away $\ket{U_L}$ states. $H_{prop,U}$ has all the same nullspace as $H_{Kitaev}$:
$$\ket{0_L 0_L}\otimes \ket{1_C0_C0_C} + (U\ket{0_L 0_L})\otimes \ket{1_C1_C0_C},$$
$$\ket{0_L 1_L}\otimes \ket{1_C0_C0_C} + (U\ket{0_L 1_L})\otimes \ket{1_C1_C0_C},$$
$$\ket{1_L 0_L}\otimes \ket{1_C0_C0_C} + (U\ket{1_L 0_L})\otimes \ket{1_C1_C0_C},$$
$$\ket{1_L 1_L}\otimes \ket{1_C0_C0_C} + (U\ket{1_L 1_L})\otimes \ket{1_C1_C0_C}$$
with the additional options of:
$$\ket{U_L 0_L}\otimes \ket{1_C0_C0_C},$$
$$\ket{U_L 1_L}\otimes \ket{1_C0_C0_C},$$
$$\ket{0_L U_L}\otimes \ket{1_C0_C0_C},$$
$$\ket{1_L U_L}\otimes \ket{1_C0_C0_C},$$
$$\ket{U_L U_L}\otimes \ket{1_C0_C0_C}$$
The third line of \eqref{hprop} requires that $H_{prop,U}$ be placed on the correct sequence of clock bits -- that if it is placed on, say, the 3rd, 5th, and 12th sites of the clock chain, that will be unsatisfiable because of the contradictory Bell pairs. The fourth line of \eqref{hprop} simply requires the right types of particles at each site.
\par Now that we have motivated the definition, we can proceed to the main result.

\begin{duplicate}[Theorem~\ref{thm:BQP1} (restated)]
\QCTS is \cls{BQP_1}-complete.
\end{duplicate}
\begin{Proof}We need to show provide a \cls{BQP_1} algorithm for deciding instances of \QCTS, and show its completeness and soundness, and then that \QCTS is \cls{BQP_1}-hard.

\proofpart{1}{\QCTS is in \cls{BQP_1}.}
\par First, a brief note about the form of the input. Providing every term of the input Hamiltonian in the standard basis would take up exponential space. We assume that the input is provided as a list of clauses (and the qudits they operate on), or a list of 5-local interactions (which might not be manifestly of the form allowed). It possible to find the set of clauses corresponding to a list of 5-local interactions by solving a system of linear equations in $O(n^5)$ time; that this preprocessing can be performed in polynomial time on a classical machine means we don't have to care about the input form. Henceforth we assume that the input is a list of $(clause,sites)$ data, indicating that clause number $clause$ is acting on the sites $sites$.

\par The following \cls{BQP_1} algorithm, we claim, solves this problem. The algorithm proceeds by identifying the structure of (possibly several) \cls{BQP_1} circuits in the problem, which is a classical operation that can be completed in polynomial time. Then it executes each circuit and verifies that the result passes.

\textbf{Algorithm 1}
\begin{enumerate}
\item For each 13-qudit, check all clauses it occurs in. Each clause will apply one of $H_L$, $H_C$, or $H_E$ to it. A given qudit should only ever have one of these applied to it, otherwise we immediately reject. Label each qudit as a logical, clock, or endpoint qudit, depending on which is applied. Any qudits that have none of these applied, have no clauses applied at all, and so can be ignored for the rest of the problem.
\item For each clock qudit and endpoint qudit, inspect all $H_{BP}$, $H_{Start,BP}$, and $H_{End,BP}$ terms applied to it. These should only ever create bell pairs between the same pairs of underlying qubits, in the $CA$, $CB$, and $EC$ subspaces. If any underlying qubit is paired with multiple others, reject.
\item Every qudit labelled as an endpoint qudit has at least one of $H_{Start}$ or $H_{End}$ applied to it. If any endpoint qudit has {\em both} $H_{Start}$ and $H_{End}$, reject. Otherwise, label it as a ``start qudit'' or ``end qudit'' accordingly, and proceed.
\item The pairs from step 2 induce a linear structure where endpoints are connected to at most one clock qudit, and each clock qudit is connected to at most two qudits on either side. Following these connections, all clock qudits and endpoint qudits can be linked into some collection of paths and cycles.
\item For any cycle from step 4 (which necessarily consists entirely of clock qudits), or any path without any endpoint qudits, assign all clock qudits the $\ket{0_C}$ state, and they can be ignored for the rest of the problem.
\item If any remaining paths have no start qudit, assign all clock qudits in that chain $\ket{1_C}$ and ignore them for the rest of the problem.
\item If any remaining paths have no end qudit, assign all clock qudits in that chain $\ket{0_C}$ and ignore them for the rest of the problem.
\item All remaining paths have at least one start qudit and end qudit. Since a single start or end can't be entangled with multiple others, it must be exactly one start qudit and end qudit -- otherwise it would have been rejected in step 2.
\item At this point we are left with a collection of clock paths, with associated unitaries from each $H_{prop,U}$ acting on logical qubits. In the case of a single clock path, this is a (ternary-logic) quantum circuit, and we need to evaluate it. In the case of multiple clock paths operating on the same qubits, we need to ensure compatibility of the two circuits, which is more complicated. We first describe the case of a single clock path, and then generalize to multiple clock paths.
\end{enumerate}
\textbf{Subroutine: Single clock path}
Every logical qubit with a $H_{Start}$ on it must be in the $\ket{0_L}$ state at $t=0$. Every other logical qubit could be in any state, including being entangled with each other; but we will see that we can assume that they all begin in the $\ket{U_L}$ state, representing an ``undefined'' logical state, without losing completeness.
Any time a qubit in the $\ket{U_L}$ state reaches a unitary gate, we can find a satisfying assignment for that circuit by terminating history there. In a fully functional circuit with unitaries $U_1, \dots U_T$ and initial state $\ket{0^n}$, the full solution to the constraint problem would be the state
$$\ket{\Psi} = \frac{1}{\sqrt{T}}\sum_{i=1}^T (U_iU_{i-1}\dots U_1\ket{0^n})\otimes \ket{Clock_i}$$
and the $H_{End}$ checks the $\ket{Clock_T}$ subspace. But if an input qudit $\ket{U_L}$ is operated on at time $t$, then a solution to the constraint problem is the state
$$\ket{\Psi} = \frac{1}{\sqrt{t}}\sum_{i=1}^t (U_iU_{i-1}\dots U_1\ket{0^n})\otimes \ket{Clock_i}$$
This is in the ground state of each $H_{prop}$, and when $H_{End}$ checks the $\ket{Clock_T}$ subspace, we are trivially in its ground state, because the projection of $\ket{\Psi}$ in that subspace is the zero vector.
\par So, if any undefined qubit is acted on by a unitary, we know we can accept, without any further checking. If no such case arises, then any solution to the QCSP must be a valid computational history. A quantum computer can execute the circuit from the known starting state $\ket{0^n}$, and measure all the check qudits, the qudits with $H_{End}$ applied. If these are all in the $\ket{0}$ state, we accept. If any are in the $\ket{1}$ state, we reject.

\textbf{Subroutine: Multiple clock paths}
If we have multiple clocks that operate on disjoint sets of qubits, then our constraint problem is completely separable, and we need only to verify each part independently. It becomes more difficult if both clock paths are operating on the same set of qubits, as there is no meaningful ordering of time. Whenever a unitary operates on a qubit and it changes state, that qubit becomes entangled with the state of that clock path. \rchange{Due to monogamy, the qubit cannot be appropriately entangled with two clock paths simultaneously.}{Coupling a logical qubit to two different clock paths does not directly violate monogamy (as a bound on entanglement measures), as the qubit will not be completely entangled with the clocks. It does, however, necessarily create a frustrated system:}

\begin{lemma}[No Shared Logical Qubits]
\label{noShare}
\rnew{Two clock paths that act on a shared logical qubit necessarily create a frustrated instance.}
\end{lemma}
\begin{proof}

\par \rnew{A valid computational history for one circuit (one clock path and its associated logical bits) is necessarily a pure state. There is a single specified input state, $\ket{0^n}$, and each applied gate is unitary; the full history can only be the pure state $\frac{1}{\sqrt{N}}\sum \ket{\psi_i}_L \otimes \ket{i}_C$. If any two of the $\ket{\psi_i}$ are not equal, the logical and clock qubits are entangled (not necessarily maximally). More precisely, in the circuits we consider here, the first time that each logical qubit is acted on, it necessarily changes. (This is a consequence of the choice of universal gate set.) This implies that each logical qubit acted on is necessarily entangled with the clock subspace.}

\par \rnew{So, consider a problem instance with two clock subsystems $C_1$ and $C_2$, that act on logical subsystems $L_1$ and $L_2$, with nonempty intersections $L_3 = L_1 \cap L_2$. Purity of computational histories implies that a frustration-free ground state restricted to $C_1 \cup L_1$ must be a pure state, as well as $C_2 \cup L_2$. $L_3$ is entangled with $C_1$, meaning that any subsystem including only one of $L_3$ or $C_1$ will necessarily be impure. $C_2 \cup L_2$ includes $L_3$ but not $C_1$, meaning it must be impure, a contradiction. So, there cannot be a frustration-free ground state.}
\end{proof}

\rnew{The remaining wrinkle here is that we might not have complete computational histories, and instead one or more of the logical qubits in each circuit may be undefined.} Thus we need to check, for each logical qubit: is it only used by one clock path? If not, is there a single state it can be in for all time in each circuit?
\par First, analyze each circuit to check for undefined variables entering unitaries. If one does, we can neglect all the remainder of that clock path, and just focus on the unitaries before the problematic one. Then check the compatibility between different clock paths:

\begin{enumerate}
\item If a qubit is only present in clauses from one clock path, this qubit is fine.
\item If this qubit is acted on by $H_{Start}$ or $H_{End}$ in multiple clock paths, but is never acted on by a gate, this qubit is fine: it can be put in a \rchange{fixed}{pure} $\ket{0_L}$ state.
\item If this qubit has no $H_{Start}$ on it, then this qubit is fine: it can be \rchange{left undefined}{put in a pure $\ket{U_L}$ state}.
\item Otherwise the qubit is initialized in multiple clock paths to $\ket{0_L}$, and subsequently acted on by unitaries by multiple clock paths. \rchange{Any unitary from our gate set acting on $\ket{0_L}$ produces a state at least a bounded distance away, leading to entanglement with that clock qubit. By monogamy, this cannot occur with both clock qubits, and we can reject this instance.}{Because it must be in a pure state when viewed together with either clock subsystem, the only possibility is that is in a pure state, and $H_{Start}$ tells us this must be $\ket{0_L}$. By Lemma \ref{noShare}, this will lead to frustration in the case where none of the unitaries are given $\ket{U_L}$ inputs. It could still be permissible if, for instance, the other inputs to unitaries are all $\ket{U_L}$. Programmatically, the way to resolve this case is to truncate each circuit at the first instance of a used undefined qubit (as time need not progress past that point). If double-used qubits remain, they really must be entangled with each of the two clock paths, and we can reject.}
\end{enumerate}
Then we can run the single clock path subroutine on each induced circuit, checking them each separately. If they all pass, then we can accept the problem.

\par \rnew{We have used classical preprocessing to separate out the system into a number of single clock-path subsystem $S_1 \dots S_\nu$ with no interactions between subsystems. We have a zero-energy ground state $\ket{\Psi}$ for the whole system iff we have zero-energy ground states $\ket{\Psi_\nu}$ for each subsytem. Each $\ket{\Psi_\nu}$ can be decomposed into a history state $\sum_{i=1}^N \alpha_i \ket{\psi_i}_L \ket{i}_C$. Each propagator at time $t$ implies $\alpha_t = \alpha_{t-1}$ and $\ket{\psi_t} = U_t \ket{\psi_{t-1}}$, up until it acts on an undefined input at some time $t_U$, where we can WLOG take $\alpha_i = 0$ for all $i \ge t_U$. If no undefined inputs are acted on, we end up with a uniform superposition of time, and $H_{End}$ validates the result.}

\par The soundness and correctness of the algorithm follow from the fact that, at each step of the preprocessing, we reject only those that we know are unsatisfiable (such as step 2), and neglect only the constraints and variables that we know are irrelvant (such as step 1). At the end we execute the circuit(s) on a quantum computer to check the results. If the instance has an exact solution, then the output qubit(s) of the circuit are precisely in the $\ket{0_L}$ state, and we accept with probability 1. This is perfect completeness.
\par Any satisfying instance of the circuit must be a valid execution history, so the only way left it could fail is if the output qubit(s) have nonzero amplitude in the $\ket{1_L}$ state. By the promise, this must be at least a $1/p(n)$ amplitude, and so our execution of the circuit has at least a $1/p(n)^2$ chance of rejecting. This is soundness.

\proofpart{2}{Quantum Clock-Horn-SAT is BQP$_1$-Hard}
Let $U_X$ be the uniform quantum circuit for a given BQP$_1$ problem. Let $m$ be the number of qubits that $U_X$ acts upon, and $k$ be the number of gates applied. Without loss of generality, we assume the gates are applied one at a time, one per time step. Then take a Quantum Clock-Horn-SAT problem on $m+k+4$ sites. $m$ will be used for the qubits, $k$ for the time steps, and 4 for the first and last clock sites and the endpoints. Label the logical qubit sites  as $Q_i$, the timestep sites as $T_i$, and the other four as $S$, $T_0$, $T_{k+1}$, and $E$. Then the clauses in the problem are:
\begin{enumerate}
\item For each $i\in [m]$: an $H_{Start}$ clause, acting on $S, T_0, T_1, Q_i$.
\item For each $t\in [k]$, where $U_t$ is the $t$th unitary of $U_X$ acting on bits $j$ and $k$: an $H_{prop,U}$ clause, acting on $Q_j, Q_k, T_{t-1}, T_t, T_{t+1}$.
\item A final $H_{End}$ clause, acting on $E, T_{k+1}, T_k, Q_1$
\end{enumerate}
Then we claim this problem, denoted $\mathcal{P}(U_X)$, is satisfiable iff $U_X \in L_{yes}$. First observe that each qudit can be identified as a logical, clock, or endpoint qudit, and so any satisfying assignment must be in the appropriate subspaces. We can now break each qudit into the subspaces of $\ket{0_CA}$, $\ket{0_CL}$, and so on. The entanglement qubits are now an entirely separate problem from the logical and ``logical clock'' qubits, and they can be easily recognized to be satisfied by Bell pairs. Since all qubits are initialized in the $\ket{0_L}$ state, it can be seen inductiviely that they can never enter the undefined $\ket{U_L}$ state without polynomially large energy penalty. \rnew{This lets us restrict to our attention to the behavior when all logical qubits are in the $\ket{0_L}$/$\ket{1_L}$ subspace. We can remove the Bell pairs, which are otherwise uncoupled, from the system.} What remains is \rnew{precisely} Kitaev's clock Hamiltonian\rchange{, which has been proved using perturbation theory to be satisfiable only if the circuit accepts with probability exponentially close to 1.}{ from \cite{Kitaev02}, up to a renaming of the basis states, which is shown there to be frustration-free iff the modelled circuit accepts. Thus the our proof of correctness requires very little direct algebraic manipulation: after throwing away the extra subspaces, we are left with the same verbatim Hamiltonian, and we can rely on that result. Thus the majority of the arguing above is about under exactly what circumstances we can discard certain subspaces or summarily reject invalid instances.}
\par \rnew{We briefly summarize Kitaev's proof here for completeness. The history state $\ket{\Psi}$ can be decomposed in the clock basis as $\ket{\Psi} = \sum_{i=0}^k \alpha_i \ket{\psi_i}_L\ket{i}_C$ with $\alpha_i$ real and positive. $\ket{\psi_0}$ is necessarily the initializing $\ket{0}^n$ state on some qubits. By construction of this hard instance, there are no undefined qubits. The propagator at time $t$ can only be in its ground state if $\alpha_{t-1} = \alpha_t$, and by induction all $\alpha$'s are equal to $1/\sqrt{k}$. Propagator $t$ also can only be in its ground state if $\ket{\psi_t}$ = $U_t \ket{\psi_{t-1}}$, implying $\ket{\psi_k} = U_{circuit} \ket{\psi_1}$. The $H_{End}$ then will produce an energy penalty unless the result of the circuit in $Q_1$ is a $\ket{0}$ state at step $t=k$.}
\par If the original problem was in the language, the \cls{BQP_1} circuit always accepts, and the QCSP is exactly satisfiable. If the original problem was not in the language, the \cls{BQP_1} circuit rejects with probability at least 1/2, and so there is no approximately satisfying assignment in the QCSP.
\end{Proof}

\section{\cls{QCMA} completeness}
There is a straightforward modification of the above that shows the existence of \cls{QCMA}-complete local Hamiltonian problems. In fact, we will construct a \cls{QCMA_1}-complete problem, then use the result of \cite{qcma1qcma} that \cls{QCMA_1}=\cls{QCMA}. The class \cls{QCMA_1} is similar to \cls{BQP_1}, except that the input is allowed to be a {\em classical} proof string. The construction remains largely the same as above, so we will not reiterate all of it, only the relevant modifications. We need to permit the logical bits to be in either the $\ket{0}$ or $\ket{1}$ state -- but not any superposition of the two, nor any entangled state.
\par One strategy to do this, if someone gave us a proof state to verify, would be to simply measure the proof state in the classical basis before proceeding with the computation. If they tried to decieve us by giving us anything other than a classical basis state, this measurement would collapse it into a classical basis state -- or, more precisely, a mixture of classical states.
\par But a constraint problem has no notion of ``measurement''. What we can do is request two copies of the state, and then perform a Swap Test on each pair of qubits. That is, for each qubit input to the problem, we take two copies of the qubit, and impose a clause
$$H_{Commit} = I - \ket{00}\bra{00} - \ket{11}\bra{11}$$
on them. This will force that measuring the second qubit in the classical basis puts the first qubit in a pure state; the first qubit cannot be in the $\ket{+}$ state or entangled with any other qubits in the input. The second qubit acts as a commitment qubit.
\par Again, we can't actually measure the second qubit, but if we make sure never to touch them for the duration of the computation, they are effectively removed from the problem, and we can trace them out to understand the computation. To prevent ourselves from using the commitment qubits in the computation, we can create a separate subspace of states for them, so that we can't get them confused with our actual logic qubits. So the construction is:
\par First, define two more basis states, $\ket{0_P}$ and $\ket{1_P}$, the commitment bit states. In addition to the $H_{Start}$ clause that forces a bit to start off as zero, we define a 4-local clause $H_{Start-Unk}$ that starts a bit off in either $\ket{0_L}$ or $\ket{1_L}$. It classically copies the logical bit to the commitment bit:
$$H_{Commit} = I - \ket{0_L0_P}\bra{0_L0_P} - \ket{1_L1_P}\bra{1_L1_P}$$
and then
\begin{align}
\label{hqcma}
H_{Start-Unk} =\,& (I - (\ket{0_{EC}0_{CA}}+\ket{1_{EC}1_{CA}})(\bra{0_{EC}0_{CA}}+\bra{1_{EC}1_{CA}}))_{12}\\
& + (I - \ket{0_{CL}}\bra{0_{CL}})\otimes (I - \ket{0_{L}}\bra{0_{L}} - \ket{1_L}\bra{1_L})_{23}\nonumber\\
& + ((I - \ket{0_{CL}}\bra{0_{CL}}) \otimes H_{Commit})_{134}\nonumber\\
&+ H_{E,1} + H_{C,2} + H_{L,3} \nonumber
\end{align}
The first, second, and fourth lines are identical to $H_{Start}$. The third line requires that, at $t=0$, the bit must be in only logical states 0 or 1, and not undefined. This is comparable to the \cls{BQP_1} construction, where we required that an input bit may only be 0. The fourth line of $H_{Start-Unk}$ requires that, at $t=0$, the bit must be equal to the commitment bit.

\par The commitment effectively turns our computation's initial state into a partial trace that removes the commitment qubits, which is a mixture of classical states. So we know that if the final computation succeeds, it is only because the input was a mixture of suceeding classical bitstrings -- which would imply that there was indeed at least one valid proof string.
\par \rnew{An equivalent setup (although not easily built as a QCSP) would be the having some $k$ extra logical qubits next to some $k < n$ we normally work with, with a requirement that the first time uses gates to copy (with \cls{CNOT}) from the originals to the extras; and the extras are never touched again. This setup is easy to analyze in the sense that we can discuss its history state. If the logical state of the verifying circuit at $t=0$ is $\ket{\psi_0} = \sum_{b=0}^{2^k} \beta_b \ket{0}^{n-k} \ket{b}$, then the initial state of the enlarged system would be}
$$\rnew{\sum_{b=0}^{2^n} \beta_b \ket{0}^{n-k} \ket{b}\ket{b}}$$
\rnew{and the full history state is}
$$\rnew{\frac{1}{\sqrt T}\sum_{t=0}^{T} U_{1\dots t}\sum_{b=0}^{2^n} \beta_b \ket{0}^{n-k}\ket{b}\ket{b} = \sum_{b=0}^{2^n} \beta_b\left(\frac{1}{\sqrt T}\sum_{t=0}^{T} U_{1\dots t} \ket{0}^{n-k}\ket{b}\right)\otimes \ket{b}.}$$
\rnew{When the first half is examined in isolation, and the half consisting of the extra copies is removed, what remains is a classical mixture of history states run on different classical input strings $\ket{b}$.}
\par \rnew{Given a verifier circuit for a \cls{QCMA} problem, we can embed it the same way we embedded in Part 2 for \cls{BQP} problems. The witness bits (which are absent in \cls{BQP}) get $H_{Start-Unk}$ instead of $H_{Start}$. Iff there is a frustration-free state, the bitstrings $\ket{b}$ in its support produce a passing output in the verifying circuit, in the same way that $\ket{0}$ produces a passing output in a \cls{BQP} circuit.} This correctness shows that the problem is \cls{QCMA_1}-hard.

\par To see that the problem is also in \cls{QCMA_1}, the proof is largely the same as the \cls{BQP_1} problem. In order to solve it in \cls{QCMA_1}, we \rchange{take a}{will accept a} classical proof string \rchange{that is the proof stings of satisfying assignments to all circuits}{built from concatenating the proof strings to each circuit chunk} in the problem\rnew{, and running each verifying circuit}. The arguments about identifying chunks, clock paths, and undefined bits all go through as before. There is the new case where a single commitment bit is a copy of several distinct input bits, or even input bits in separate circuits. But this merely imposes the restriction that the classical proof string has equal bits in those two locations, which does not make the problem any harder to verify\rnew{: when a proof string is provided, check for any commitment clauses that share the same logical qubits, and if those two bits in the proof string differ, reject. Otherwise, the problem is identical to the problem in \cls{BQP}, except with this alternate provided input state, and the same algorithm can be used with this alternate initial state.} This modified verification algorithm implies that the problem is in \cls{QCMA_1}, and so is \cls{QCMA_1}-complete.

\section{\cls{coRP} completeness}
The reduction to a \cls{coRP}-complete problem is even simpler, since we build on the idea of proving \cls{MA}-hardness from \cite{Brav09}:

\begin{displayquote}
Any classical \cls{MA} verifier $V$ can be transformed into a quantum verifier $V'$ which uses a quantum circuit $U$ involving only classical reversible gates (for example, the 3-qubit Toffoli gates) together with ancillary states $\ket{0}$, $\ket{+}$, and measures one of the output qubits in the $\ket{0}$, $\ket{1}$ basis.
\end{displayquote}

Note that \cls{MA} is just the version of \cls{coRP} that is allowed to have a proof string provided. (Really, this is \cls{MA_1}; but \cls{MA=MA_1}.) Thus, we can modify the \cls{BQP_1}-complete problem as follows:
\begin{itemize}
\item Replace the universal quantum 2-qubit gate set with a universal classical reversible 3-bit gate set. This changes from $H_{prop}$ from being 5-local to being 6-local.
\item In addition to $H_{Start}$, we have a $H_{Start-Rand}$. They are identical except in that $H_{Start-Rand}$ initializes in the $\ket{+}$ state instead of $\ket{0}$.
\end{itemize}
The remainder of the proof holds just as before. The preprocessing is doable classically, and so can be executed by even a simple \cls{coRP} machine. The final circuit to evaluate is a classical probabilistic verifier and so can be done by the \cls{coRP} machine as well.
\par\rnew{To be precise, the \cls{BQP} proof is completely agnostic to the gate set, in the sense that it shows the clauses built from gate set $\mathcal{G}$ is a complete problem for computations with the gate set $\mathcal{G}$. It is also agnostic to how the initial state is specified, as long as there is no entanglement. Separately, it was shown in \cite{Brav09} that the gate set of the 3-qubit Toffoli gates is universal for classical probabilistic computation, given access to initial states $\ket{0}$ and $\ket{+}$. Since the \cls{BQP}-completeness proof is agnostic to these modifications, the resulting set of clauses is complete for \cls{coRP}.} 
\par In \cite{PCP2019} a coRP-complete problem involving local Hamiltonians is also constructed, using their notion of ``pinned'' Hamiltonians: these are problems with the promise that the ground state $\ket{\psi}$ has nonzero overlap with the all-zero state, i.e. $\braket{0|\psi}>0$. This pinning promise allows the verifier to assume that $\ket{0}$ is a functioning 'witness', reducing the complexity from MA to coRP. The ``pinning'' promise cannot be expressed in terms of local constraints, though, so it is not a QCSP in the sense we have defined it.

\section{Weak QCSPs}
As noted in section 3.1, there is the detail of gate sets and exact results. For classical computations, there are finite universal gate sets (such as $\{CCNOT\}$). Even for probabilistic computations, we only need uniformly random bits (or a 50-50 bit flip gate) to build a robust definition of \cls{BPP}, \cls{RP}, \cls{MA}, \cls{AM}, \cls{PP}... and so on -- as uniformly random bits can be used to closely approximate other probabilities, and we only need certain bounds on the probability of a given trial.
\par For quantum circuits, phases can cancel out exactly, so that a single phase gate or qubit distribution cannot suffice for exact computations. Since different circuits might need all variety of phases, we cannot keep perfect exactness with a finite set while allowing all the circuits we might want; at the same time, allowing an infinite set of gates creates issues such as uncomputable amplitudes, and raises questions of how we are to encode the circuit.
\par The Solovay-Kitaev theorem gives a weaker, but arguably more natural, notion of universal gate set, one that allows us to approximate a given gate to exponentially good accuracy in polynomial time. This is not useful to use in studying constraint problems as we have defined them together, as we have required that all ground states be exact. To remedy this, we present the notion of a {\em weak quantum CSP}, which will allow exponentially small errors (energy) in the ground state.
\begin{definition}[Weak Quantum CSP]
 A weak QCSP has a domain size $d$, a set of clauses $C = \{\mathcal{H}_i\}$, and constants $a$, $b$, and $c$, with $b > a > 0$. Each clause $\mathcal{H}_i$ of arity $m$ is a Hermitian projector on $(\mathbb{C}^d)^{\otimes m}$.
 \par An instance of this weak QCSP is given by an integer $n$ indicating the number of $d$-qudits, and a list of clauses that apply to $n$ qudits. The instance is satisfiable if the ground state energy is less than $a/n^c$, and is unsatisfiable if the ground state energy is greater than $b/n^c$, and we are promised one of these is the case. (So, weak QCSPs are a class of promise problems.)
\end{definition}
\par It is worth emphasizing that the constants $a$, $b$, and $c$ are not allowed to vary with the instance, but are rather part of the language itself. This restriction means that every weak QCSP is in \cls{QMA}, because we only need to measure the energy to within a polynomial gap: after $O(n^b)$ measurements we can have a probability of error bounded away from $1/2$. We choose to fix these constants for a whole constraint problem class, in contrast to $k\textsc{-Local-Hamiltonian}$ in \cite{Kemp03}, where they are parameters of an instance.
\par At the same time, weak QCSPs are generally independent of gate set
. While strong QCSPs naturally describe complexity classes with one-sided error, weak QCSPs naturally describe complexity classes with two-sided bounded error. If we have a quantum algorithm that uses a gate set $\mathcal{G}_1$ to solve a problem with two-sided bounded error in polynomially many gates, we can also solve with two-sided bounded error in polynomially man gates using any other universal gate set $\mathcal{G}_2$. This is shown by applying the Solovay-Kitaev theorem, to simulate the first gate set using the second one. The Solovay-Kitaev theorem provides $O(2^{-k})$ precision per gate with only $O(k)$ times more gates. If the original algorithm uses $f(n)$ gates, we can always solve the weak QCSP in $O(f(n) \log(f(n))$ gates with bounded error. The fact that we don't have to worry about exact gate set will make several complexity classes more natural, at the expense of a more complicated notion of ``problem''.

\subsection{Results on weak problems}
First, it is straightforward to verify that there are still \cls{P}, \cls{NP}, \cls{MA}, and \cls{QMA} complete problems among weak QCSPs. \cls{P}-complete and \cls{NP}-complete problems like Horn-SAT and 3-SAT are described by commuting projectors, so the ground state energy of an unsatisfiable instance is always at least 1. Thus taking $a=1/3,\,b=2/3,c=0$ suffices to formulate those problems as weak QCSPs.
\par To see that we still have a \cls{QMA}-complete problem, refer to the specifics of the proof that $3\textsc{-Local-Hamiltonian}$ is \cls{QMA}-complete in \cite{Kemp03}: their construction yields a problem whose ground state energy is below $\frac{c_1\varepsilon}{n}$ if satisfiable, or above $\frac{c_2}{n^3}$ is unsatisfiable, where $\varepsilon$ is the algorithm's allowed probability of error. By running the algorithm $O(n)$ times in series, $\varepsilon$ becomes $O(2^{-n}) < \frac{1}{n^2}$, and so it suffices to take $a=c_1,\,b=c_2,\,c=3$. Then their construction permits the rewriting of any \cls{QMA} circuit into an instance of this weak QCSP. We could also do this with \cls{QMA_1}-complete \cls{3-QSAT} to arrive at \cls{QMA}-complete weak QCSP.
\par This shows the scheme by which a $1/poly(n)$ gap promise translates into the weakness of the QCSP. \cite{Brav09} have the same gap for their their \cls{MA}-complete problem, and so it also can be described as weak QCSP. Interestingly, in this case, \cls{MA}$=$\cls{MA_1}, and so the complexity class does not change; the same holds for \cls{QCMA}$=$\cls{QCMA_1}.
\par The other two cases we gave above, \cls{coRP} and \cls{BQP_1} also all have polynomially small gaps, and are designed to emulate a circuit that accepts perfectly on accepting instances. If we instead allowed our original circuit to have a polynomially small {\em two}-sided error, we could build a weak QCSP instead, that is complete for the corresponding two-sided error complexity class. If the original QCSP had a minimum unsatisfiable ground-state energy of $O(n^{-p})$, then whatever circuit we are embedding as QCSP, let us repeat it enough times that its error is $O(n^{-p-1})$. Then choosing $c = p+0.5$ and taking any positive $a$ and $b$ describe it as a weak QCSP. This gives the following result.
\begin{corollary}
There are weak QCSPs that are complete for the classes \cls{BPP} and \cls{BQP}.
\end{corollary}
This means that any putative dichotomy-like theorem for weak QCSPs would need to at least account for the seven cases of \cls{P}, \cls{BPP}, \cls{NP}, \cls{MA}, \cls{BQP}, \cls{QCMA}, \cls{QMA} -- or show that some pair of these are equal. The insensitivty to choice of gate set could make this type of result more appealingly natural, and indeed more physical, as no real-world gate set can be realized with zero error.

\section{Universality of qubits for QCSPs}
We can show that any QCSP can be reduced to another problem using only only qubits, with little computational power. This may initially sound unsurprising, as operations on qubits are certainly universal for quantum computation. But it is perhaps surprising, in light of the fact that the analogous statement is believed to be {\em false} in the classical world!
\par Among classical constraint problems, it is believed that distinct complexity classes arise for different size domains. For boolean constraint problems, it is known\cite{binCsp2009} that every problem is either \cls{coNLOGTIME}, \cls{L}-complete, \cls{NL}-complete, $\oplus$\cls{L}-complete, \cls{P}-complete, or \cls{NP}-complete. This is a refinement of the dichotomy theorem specialized to boolean problems, as the first five classes are all contained within \cls{P}. Among ternary constraint problems though, there are new classes that appear, such as \cls{Mod_3L}-complete, which are not expected to be equal to any of the six the previously listed.\cite{modKL}\cite{polymorphisms}
\par When constructing a circuit for a quantum computer, we can emulate a $d$-qudit with a $\lceil \log_2(d)\rceil$ qubits, and carry out unitaries on those qubits. We can certainly try the same thing for a QCSP, turning (for instance) each 4-qudit into 2 qubits, and a $k$-local clause becomes $2k$-local. The issue arises that we cannot ensure that the $2k$-local clauses are applied to qubits in a consistent fashion. One clause might treat a particular qubit as ``bit 1'' of an 4-qudit, while another clause might use that same qubit as ``bit 2''. This would lead to constraints that were previously unrealizable. Clauses could also ``mix and match'', combining ``bit 1'' from one 4-qudit with ``bit 2'' from another 4-qudit. The exact same problems exist in the classical setting.
\par In the quantum world, we can fix this, again by using monogamy to bind together our constituent qubits into ordered, entangled larger systems. Each clause in the resulting problem is given a projector that forces this particular ordering of qubits, and any two clauses that try use the same qubits in multiple ways are frustrated. Formally,

\begin{duplicate}[Theorem~\ref{thm:q2Informal} (formal)]
\label{thm:q2formal}
For any QCSP $\mathcal{C}$ on $d$-qudits, there is another QCSP $\mathcal{C}'$ on qubits, and \cls{AC^0} circuits $f$ and $g$, such that $f$ reduces $\mathcal{C}$ to $\mathcal{C'}$, and $g$ reduces $\mathcal{C}'$ to $\mathcal{C}$. If $\mathcal{C}$ is $k$-local, then $\mathcal{C'}$ can be chosen to be $4\cdot 2^{\lceil \log_2\left(\lceil \log_2(d)\rceil \right) \rceil}k $ local (that is, $O(\log(d))$ times larger.)
\end{duplicate}

\begin{proof}
First we will show that for any $d$-qudits, we can reduce to 4-qudits; after that we will reduce to qubits. Finally we show that the reduction is in \cls{AC^0}.
\par We will view 4-qudit as the product of a ``data'' qubit and an ``entanglement'' qubit. A $d$-qudit will be replaced by $\rnew{n=\,}\lceil \log_2(d)\rceil$ many 4-qudits, and the state of the $d$-qudit will be encoded in the product of all the data qubits. If $d < 2^{\lceil \log_2(d)\rceil} \rnew{\,= 2^n}$, that is, if $d$ is not exactly a power of 2, we will have a Hamiltonian term $T_1$ in our clauses to ensure that the last $2^n - d$ states are not used. Acting on entanglement subspaces of the 4-qudits, consider a term $T_2$ whose nullspace consists of just the vector
$$\left(1\otimes X^\theta \otimes X^{2\theta} \otimes \dots X^{n\theta}\right)\frac{\ket{0}^n + \ket{1}^n}{\sqrt{2}}$$
where $\theta = \frac{1}{2(n+1)}$. This is a kind of GHZ state, which uses a slightly different pair of basis states (instead of just $\ket{0}$ and $\ket{1}$) on each separate qubit. Any bipartition of this state is impure, but since $T_2$ has a one-dimensional nullspace, it cannot be satisfied by any mixed state. Thus the sum of two $T_2$ on any two overlapping but distinct sets of 4-qudits will be frustrated. If two copies of $T_2$ act on the same 4-qudits in a different order, they will apply the wrong angles $X^{k\theta}$ at those sites, and the ground states do not align -- also leading to frustration.
\par Each clause $H$ of $\mathcal{C}$ is mapped to a new clause $H'$ that acts as $H$ on the data subspaces of each set of 4-qudits; that has $T_1 = 1 - \sum_i^d \ket{i}\bra{i}$ on each clumping of 4-qudits, to ensure that only the first $d$ states are used; and $T_2$ on each clumping of 4-qudits, to ensure that they will only be clustered to each other and in a particular order.
\par Then we want to reduce this from 4-qudits to qubits. Consider the following Hamiltonian on 4 qubits:
$$H_{4\to 2} = 1 - \ket{\psi_1}\bra{\psi_1} - \ket{\psi_2}\bra{\psi_2}  - \ket{\psi_3}\bra{\psi_3}  - \ket{\psi_4}\bra{\psi_4}$$
$$\ket{\psi_1} = \frac{1}{2}\left(\frac{3}{5}\ket{0000} - \frac{4}{5}\ket{0001} + \ket{0100} + \ket{1010} + \frac{8}{17}\ket{1100} + \frac{15}{17}\ket{1111}\right)$$
$$\ket{\psi_2} = \frac{1}{2}\left(\frac{4}{5}\ket{0000} + \frac{3}{5}\ket{0001} - \ket{0110} + \ket{1001} + \frac{20}{29}\ket{1101} + \frac{21}{29}\ket{1110}\right)$$
$$\ket{\psi_3} = \frac{1}{2}\left(\frac{5}{13}\ket{0010} + \frac{12}{13}\ket{0011} - \ket{0111} + \ket{1000} - \frac{21}{29}\ket{1101} + \frac{20}{29}\ket{1110}\right)$$
$$\ket{\psi_4} = \frac{1}{2}\left(\frac{-12}{13}\ket{0010} + \frac{5}{13}\ket{0011} - \ket{0101} + \ket{1011} - \frac{15}{17}\ket{1100} + \frac{8}{17}\ket{1111}\right)$$
Each $\ket{\psi_i}$ is orthonormal, so $H_{4\to 2}$ has a nullspace of dimension four. By inspecting the 840 distinct ways to apply two copies of $H_{4\to 2}$ to seven qubits, it can be checked that each sum will have a ground state above zero, except for the case where they are applied in the same way. This is a kind of ``uniqueness'' property that could be very loosely interpreted as monogamy for whole subspaces, instead of just one state.
\par By counting dimensions, one can check that this property is generic: it would hold almost always for any four random vectors. Unfortunately, for any simple and clean expressions one would write down, it would lack this property by one symmetry or another. This is why the simplest construction readily available, given above, is actually quite ugly.
\par Given a problem on 4-qudits, we can replace each 4-qudit with a collection of four qubits. Each clause is modified to act on the $\ket{\psi_i}$ basis of qubits instead of $\ket{i}$ basis of the 4-qudits. Then, for each 4-qudit that the clause acted on, we add a copy of $H_{4\to 2}$. The above uniqueness property ensures that no other clause can act on the same qubits in any other order, or mix them with any other set of qubits.
\par In the $\mathcal{C}'$ QCSP, any problem where qubits are mixed or applied in inconsistent orders, can immediately be rejected. Some qubits may not be acted on by any clause, and so not correspond to a $d$-qudit in $\mathcal{C}$, but then those qubits can simply be ignored. This leaves us with only correctly grouped qubits, in a certain subspace, that thus function equivalently to the $d$-qudits.
\par Combined, this gives a faithful reduction from $d$-qudits to qubits, and back again. It turns a $k$-local Hamiltonian into a $4\lceil \log_2(d)\rceil k$ local Hamiltonian. It remains to check the complexity of the reductions.
\par In the above description, the expansion factor is $4\lceil \log_2(d)\rceil$. To get a low circuit complexity, we want the expansion to be a power of two. Thus we round this up to $4\cdot 2^{\lceil \log_2\left(\lceil \log_2(d)\rceil \right) \rceil}$, which will denote $x$; the QCSP $\mathcal{C}'$ is augmented through just adding more qubits to increase the subspace dimension, and then $T_1$ is modified again to prevent those states from being occupied.
\par An instance of a $k$-local QCSP $\mathcal{C}$ can be written down as a list of integers, each given by an integer clause type, and $k$ many integers representing the qudits they act on. The clause types of $\mathcal{C'}$ are in one-to-one correspondence with the clause types of $\mathcal{C}$, so those integers remain unchanged. Each clause acting on qudit $i$ now instead acts on qudits $[xi,xi+1,\dots xi+(x-1)]$ in that order. Thus a reduction circuit $f$ only needs to be able to replicate an integer several times, multiply by a constant power of 2, $x$, and add a number $i \in [0,x)$. This is in $AC^0$ (in fact it requires no gates at all).
\par For a circuit $g$ to convert back from $\mathcal{C'}$ to $\mathcal{C}$, we need to map qubit numbers back to qudits numbers, and check that no qubits are used in inconsistent fashion. For each collection of qubits $[a_1,a_2,\dots a_x]$ that a clause in $\mathcal{C}'$ is acting on, we can map that to the $d$-qudit number $a_1$. Thus, many qudit numbers will go unused, but this doesn't affect the correctness: as long as all collections use the same numbers in the same order, they will all be mapped to $a_1$. To check that all qudits are used in a consistent fashion, we need to check for each pair of collections $([a_1,\dots a_x], [b_1,\dots b_x])$ that they do not use qubits in inconsistently. Logically, this reads:
$$\big(((a_i = b_i) \wedge (a_j = b_j)) \vee ((a_i \neq b_i) \wedge (a_j \neq b_j))\big) \wedge (a_i \neq b_j) \wedge (a_j \neq b_i)$$
and this must be checked for every collection, for every $i$ and $j$, and then combined by an unbounded fan-in AND. The integer equalities $a_i = b_i$ and $a_i \neq b_i$ can be evaluated with unbounded fan-in AND and OR respectively. This check is all in \cls{AC^0}. If the check fails, the circuit outputs some fixed clause(s) that are unsatisfiable, otherwise it outputs a repetition of the first clause. This \cls{AC^0} circuit checks that the qudits are used consistently, and if they are, gives an equivalent instance in the original language.
\end{proof}
If we don't care about having the reductions be in \cls{AC^0}, and instead allow \cls{P}-reductions, then $4\lceil \log_2(d)\rceil k$ locality suffices. This reduction is optimal within a factor of 4, in the sense that encoding one $d$-qudit in several qubits requires at least $\lceil \log_2(d)\rceil$ many qubits. In section 4 we showed that there is a 5-local 13-qudit problem that is \cls{BQP}-complete. Together with Theorem~\ref{thm:q2Informal}, we have:
\begin{corollary}
There is a \cls{BQP}-complete QCSP on qubits, with 80-local interactions.
\end{corollary}
In practice the locality could be reduced quite a bit, probably below 20 with some work.

\section{Future directions}
The seven complexity classes that are known to occur as strong QCSPs are now, in rough order of difficulty:
\begin{enumerate}
\item \cls{P}: Classical, no proof, deterministic checks.
\item \cls{coRP}: Classical, no proof, probabilistic checks.
\item \cls{NP}: Classical, classical proof, deterministic checks.
\item \cls{MA}: Classical, classical proof, probabilistic checks.
\item \cls{BQP_1}: Quantum, no proof, probabilistic checks.
\item \cls{QCMA}: Quantum, classical proof, probabilistic checks.
\item \cls{QMA_1}: Quantum, quantum proof, probabilistic checks.
\end{enumerate}
Are there obvious omissions we should expect to look in, or does this list seem complete? It seems natural in one way: we can have classical or quantum verifiers; we can have no proofs, classical proofs, or quantum proofs; and we can have deterministic or probabilistic verification. This would produce 12 classes in total. But we cannot have a quantum proof for a classical verifier, bringing us down to 10 classes. And it is very hard to force deterministic verification on a quantum verifier; \cls{EQP} is a difficult class to study. Any construction of an \cls{EQP}-complete constraint problem would likely require knowledge of particular forms of circuits that are powerful enough to capture the full power of \cls{EQP}, while constrained enough to guarantee that they always produce deterministic results.
\par It is clear, though, that the set of gates can be freely exchanged. There is the question if there are gate sets that are more powerful than classical computation but weaker than universal quantum computation. This is plausibly the case for nonstandard models of computation, such as sampling problems \cite{Boul18}\cite{Knil02} or one-clean-qubit models\cite{Fujii18}\cite{Suss20}, but is less clear for the standard quantum circuit model. But if it was discovered that there was such a set of intermediate-power gates, we would immediately have a corresponding constraint problem class which captures its difficulty.
\par The same questions linger for weak QCSPs. \cls{EQP}-complete problems are extremely unlikely to exist among weak QCSPs, as weak QCSPs do not depend on gate sets and \cls{EQP} problems genearlly do. We can also ask about exotic gate sets yielding two-sided error complexity classes as weak QCSPs.
\par The constructions given herein is certainly very large, both in arity of the clauses and local dimension of the qudits. It is natural to ask if there are smaller and more natural constraint problems that could realize the classes we discussed. In particular, finding a 3- or even 2-local constraint would be exciting. This does not seem implausible given that 2-local  Hamiltonian problems are in general \cls{QMA}-hard. Ideally we would even have something analogous to Theorem 8 for locality reduction, that all QCSPs could be made 2-local using high-dimensional qudits. In classical CSPs these are well-studied under the name of {\em binary} constraint problems.\cite{COHEN201912}
\par Finally, regarding \cls{coRP} and \cls{BPP}, it is somewhat surprising that these easy classical complexity class arises from a purely quantum Hamiltonian. It would be good to reformulate the problem as closely as possible in purely classical terms. Aharanov and Grilo have recently performed a reformulation like this for \cls{MA}, transforming the language of stoquastic Hamiltonian constraints into a much simpler classical problem\cite{Ahar20}.

\renewcommand{\abstractname}{\large Acknowledgements}
\begin{abstract}
I would like to thank Bela Bauer, Scott Aaronson, Alex Grilo, and Sevag Gharibian for useful feedback and discussions. This research was funded by Microsoft Station Q.
\end{abstract}

\printbibliography

\typeout{get arXiv to do 4 passes: Label(s) may have changed. Rerun}
\end{document}